\documentclass[letterpaper, 10pt, conference]{ieeeconf} 
\IEEEoverridecommandlockouts
\usepackage{graphicx}
\usepackage{url}
\usepackage{xcolor}
\usepackage{amsmath}
\usepackage{mathtools}
\usepackage{tikz}
\usepackage{verbatim}
\usepackage{caption}
\usepackage{subcaption}
\usepackage{pgfplots}
\pgfplotsset{compat=1.10}
\usepackage{graphicx}
\usepackage{epstopdf}
\usepackage{amssymb}
\usepackage{relsize}
\usepackage[textwidth=3.8em,textsize=scriptsize,disable]{todonotes}
\usepackage{algorithm}
\definecolor{Gray}{gray}{0.90}
\usepackage{colortbl}

 \usepackage{url}
 \usepackage{hyperref}
\usepackage{algorithm}
\usepackage{algorithmicx}
\usepackage{algcompatible}
\usepackage{algpseudocode}
\usepackage{color}
\usepackage{blkarray}
\usepackage{varwidth}
\usepackage{cite}
\usepackage{setspace}
\newtheorem{theorem}{\bf{Theorem}}[section]

\newtheorem{prop}[theorem]{\bf{Proposition}}

\newenvironment{definition}[1][Definition]{\begin{trivlist}
\item[\hskip \labelsep {\bfseries #1}]}{\end{trivlist}}

\usepackage{amssymb}
\usepackage{url}

\newcommand{\calD}[0]{\mathcal{D}}

\newcommand{\calM}[0]{\mathcal{M}}

\newcommand{\calP}[0]{\mathcal{P}}
\newcommand{\calC}[0]{\mathcal{C}}

\newcommand\encircle[1]{%
  \tikz[baseline=(X.base)] 
    \node (X) [draw, shape=circle, inner sep=0] {\strut #1};}
\newtheorem{prob}{\textbf{Problem}}

\newtheorem{interpret*}{Interpretation:\\}

\usepackage{setspace}
\usepackage{relsize, scalerel}
\begin{document}
\title{Controllability Backbone in Networks}
\author{Obaid Ullah Ahmad, Waseem Abbas, 
 and Mudassir Shabbir
\thanks{Obaid Ullah Ahmad is with the Electrical Engineering Department at the University of Texas at Dallas, Richardson, TX. Email: Obaidullah.Ahmad@utdallas.edu.}
\thanks{Waseem~Abbas is with the Systems Engineering Department at the University of Texas at Dallas, Richardson, TX. Email: waseem.abbas@utdallas.edu.}
\thanks{Mudassir~Shabbir is with the Computer Science Department at the Vanderbilt University, Nashville, TN. Emails: mudassir.shabbir@vanderbilt.edu.}
}

\maketitle
\begin{abstract}
This paper studies the controllability backbone problem in dynamical networks defined over graphs. The main idea of the controllability backbone is to identify a small subset of edges in a given network such that any subnetwork containing those edges/links has at least the same network controllability as the original network while assuming the same set of input/leader vertices. We consider the strong structural controllability (SSC) in our work, which is useful but computationally challenging. Thus, we utilize two lower bounds on the network's SSC based on the zero forcing notion and graph distances. We provide algorithms to compute controllability backbones while preserving these lower bounds. We thoroughly analyze the proposed algorithms and compute the number of edges in the controllability backbones. Finally, we compare and numerically evaluate our methods on random graphs.
\end{abstract}

\section{Introduction}
Network structure profoundly influences the dynamical behavior of networked multiagent systems. For instance, network controllability, connectivity, robustness to failures, information dissemination, and influence evolution in networks rely on the underlying network topology~\cite{mesbahi2010graph}. Therefore, any changes to the network's structural organization, such as adding or removing links between agents, may alter the system-level properties of the network, which could be either beneficial or detrimental. Thus, for a survivable network design and avoid the deterioration in the desired network behavior, a practical approach is to identify a sparse subnetwork (or backbone) whose maintenance would guarantee the preservation of the desired network property in the face of modifications. For example, to maintain connectivity, preserving edges in the minimum spanning tree ensures a path between every pair of agents. Similarly, in communication infrastructure networks, connected dominating sets are used to identify the minimum number of agents necessary to form the backbone network~\cite{yu2013connected}. 

This paper studies the \emph{controllability backbone} problem in a networked dynamical system defined over a graph $G=(V,E)$. Network controllability concerns the ability to manipulate the agents within a network as desired through external control signals injected via a subset of agents called \emph{input agents} or \emph{leaders}. The network controllability depends on the choice of leaders $V_\ell\subseteq V$ and the interconnections between agents~\cite{pasqualetti2014controllability,summers2015submodularity,becker2020network}. Moreover, the network controllability may deteriorate if the connections/edges between agents change~\cite{abbas2020tradeoff,chanekar2019network,mousavi2017structural,xiang2019advances}. The main idea of the controllability backbone is to determine a small subset of edges $E_B\subseteq E$ such that \emph{any} subnetwork of $G$ containing $E_B$ has at least the same network controllability as $G$ with the same leaders. In other words, maintaining $E_B$ implies that the minimum network controllability is preserved despite edge modifications.

We consider the \emph{strong structural controllability (SSC)} for the backbone problem. SSC is advantageous as it depends on the edge set $E$ and not on the edge weights (which represent the coupling strengths between vertices and often are not precisely known). However, determining the SSC of a network is a challenging computational problem~\cite{chapman2013strong,mousavi2017structural,shabbir2022computation}. So a typical approach is to obtain tight lower bounds. Therefore, we aim to identify a controllability backbone for a given network $G = (V,E)$ and leader set $V_\ell$, where the backbone preserves a tight lower bound on the network's SSC. As for the SSC lower bounds, we consider two widely used bounds based on the zero forcing sets and distances in graphs~\cite{monshizadeh2014zero,yaziciouglu2016graph,zhang2013upper}. Our main contributions are as follows:

\begin{enumerate}
\item We present a novel approach to identifying a sparse subgraph in a graph that guarantees the same level of controllability (SSC) as the original graph. We call this subgraph the \emph{controllability backbone} (Section~\ref{sec:prelim}).

\item We provide a polynomial algorithm to compute a minimum controllability backbone, which preserves a lower bound on the network's SSC based on zero forcing sets in graphs (Section~\ref{sec:zfs}). 

\item Additionally, we consider a distance-based lower bound on SSC and compute a controllability backbone preserving the distance bound. We derive tight bounds on the number of edges in the distance-based 
 backbone (Section~\ref{sec:dist}).

\item Finally, we illustrate our results and compare different controllability backbones (Section~\ref{sec:comp}). 
\end{enumerate}

There are previous works dealing with the densification problem, i.e., how can we add edges to a graph while maintaining its controllability (e.g., \cite{abbas2020improving,mousavi2020strong})? In contrast, this paper studies an inverse, i.e., the sparsification problem, to identify a small subset of crucial edges whose existence within any subgraph guarantees the same controllability as the original graph. While some studies have considered identifying edges whose removal from the graph does not deteriorate the network controllability of the remaining graph (e.g., \cite{mousavi2017robust,rahimian2013structural,jafari2010structural,pu2012robustness}), our problem setup is distinct. We require that \emph{any} subgraph containing the backbone edges be at least as controllable as the original graph, resulting in a more general problem formulation. Furthermore, our formulation considers the concept of strong structural controllability, which adds to its generality.

The rest of the paper is organized as follows: Section~\ref{sec:prelim} introduces preliminaries and sets up the controllability backbone problem. Section~\ref{sec:zfs} reviews the zero forcing ideas and then studies the zero forcing-based controllability backbone problem. Section~\ref{sec:dist} describes the distance-based bound on network SSC and then employs it to compute the distance-based backbone. Section~\ref{sec:comp} compares the two controllability backbones and numerically evaluate the proposed methods. Finally, Section~\ref{sec:conclusion} concludes the paper. 
\section{Preliminaries and Problem Formulation}
\label{sec:prelim}
\subsection{Notations and System}
An undirected graph $G=(V, E)$ models a multiagent network. The vertex set $V$, and the edge set $E\subset V\times V$ represent agents and interactions between them, respectively. The edge between vertices $u$ and $v$ is denoted by an unordered pair $(u,v)$. The \emph{neighborhood} of $u$ in graph $G$ is the set $\mathcal{N}_G(u) = \{v\in V:\; (u,v) \in E\}$ and the \emph{degree} of $u$ is $\deg(u) = |\mathcal{N}_G(u)|$. 
A \emph{path} $P$ in a graph $G$ is defined as a sequence of vertices $(v_1, v_2, v_3, \cdots, v_k)$, where $v_1, v_2, v_3, \cdots, v_k$ are distinct vertices in the graph, and for every $i$ from $1$ to $k-1$, there exists an edge between $v_i$ and $v_i+1$.
The \emph{distance} between vertices $u$ and $v$, denoted by $d(u,v)$, is the number of edges in the shortest path between $u$ and $v$. A graph $\hat{G} = (V, \hat{E})$ is a \emph{subgraph} of $G = (V, E)$, denoted by $\hat{G} \subseteq G$, if $\hat{E} \subseteq E$, and $G$ will be a super graph of $\hat{G}$.

We consider a network of $n$ agents, denoted by $V=\{v_1,v_2,\cdots,v_{n}\}$, of which $m$ are input/leader vertices, which are represented by $V_\ell=\{\ell_1,\ell_2,\cdots,\ell_{m}\}\subseteq V$, and the rest are followers. We consider the following liner time-invariant system on $G$.
\begin{equation}
\label{eq:system}
    \dot{x}(t) = Mx(t) + Hu(t).
\end{equation}
Here, $x(t)\in\mathbb{R}^{n}$ is the state vector and $u(t)\in\mathbb{R}^m$ is the external input injected into the system through $m$ leaders. $M\in\mathcal{M}(G)$ is the system matrix, where $\mathcal{M}(G)$ is a family of symmetric matrices associated with $G$ defined as:
\begin{equation}
\label{eq:mtx_family}
\begin{split}
\mathcal{M}(G) = \{M\in\mathbb{R}^{n\times n}\;& : \;M = M^\top, \text{ and for }i\ne j,\\
& M_{ij} \ne 0 \Leftrightarrow (i,j) \in E(G)\}.
\end{split}
\end{equation}
The matrix $H\in\mathbb{R}^{n\times m}$ in \eqref{eq:system} is the input matrix, such that $H_{ij}=1$, if $v_i = \ell_j$; and $0$ otherwise.
We note that the input matrix $H$ is defined by the selection of leader agents. Here, $\mathcal{M}(G)$ denotes a broad class of system matrices defined on graphs, including the adjacency, Laplacian, and signless Laplacian matrices. 

\subsection{Strong Structurally Controllable Networks} The system \eqref{eq:system} is \emph{controllable} if there exists an input $u(t)$ that can drive the system from an arbitrary initial state $x(t_0)$ to any desired state $x(t_f)$ in a finite amount of time. If the system is controllable for a given system and input matrices, we say that $(M,H)$ is a \emph{controllable pair}. Moreover, $(M,H)$ is a controllable pair if and only if the controllability matrix $\mathcal{C}(M,H)\in\mathbb{R}^{n\times nm}$ is full rank, i.e., $\texttt{rank}(\mathcal{C}(M,H)) = n$. The controllability matrix is defined as: 
\begin{equation}
\label{eq:Gamma}
\mathcal{C}(M,H) = \left[\begin{array}{lllll}
    H & MH & M^2H & \cdots & M^{n-1}H  \\
\end{array}\right].
\end{equation}

\begin{definition} (\emph{Strong Structural Controllability (SSC))} A graph $G=(V,E)$ with a given set of leaders $V_\ell\subseteq V$ (and the corresponding $H$ matrix) is \emph{strong structurally controllable} if and only if $(M,H)$ is a controllable pair \emph{for all} $M\in\mathcal{M}(G)$. 
\end{definition}
If the network $G$ is strong structurally controllable for a given set of leaders, then the rank of the controllability matrix does not depend on the edge weights (as long as they satisfy \eqref{eq:mtx_family}). For the rest of the paper, we refer to strong structural controllability simply as \emph{controllability}. The dimension of strong structurally controllable subspace, denoted by $\gamma(G, V_\ell)$, is the smallest possible rank of the controllability matrix under feasible weights.

\begin{definition}(\emph{Dimension of SSC)} For a fixed leader set $V_\ell$, the dimension of strong structurally controllable subspace, denoted by $\gamma(G, V_\ell)$, is the smallest possible rank of the controllability matrix over all $M\in \mathcal{M}(G)$, i.e.,
\begin{equation}
\label{eq:min_rank}
\gamma(G, V_\ell) =  \min_{M \in \calM(G)} (\texttt{rank}(\calC(M, H))).
\end{equation}    
\end{definition}

$\gamma(G, V_\ell)$ quantifies `how much' of the network $G$ can always be controlled through the leaders $V_\ell$.
\subsection{Controllability Backbone Problem}
\label{sec:Ctrb_backbone}
We are interested in identifying a small subset of edges among vertices within a network that would maintain its strong structural controllability in its subnetworks. This entails identifying the sparsest subgraph, referred to as the \emph{controllability backbone}, that guarantees at least the same level of controllability as the original network in any subnetwork that encompasses the controllability backbone. In essence, the controllability backbone represents the minimum structure that must be preserved within the network to ensure its minimum controllability despite structural perturbations.

\begin{definition} 
(\emph{Controllability Backbone $B$)} For a given $G = (V,E)$ and leaders $V_\ell\subseteq V$, the controllability backbone (or simply \emph{backbone}) $B = (V, E_B)$, is a subgraph of $G$ with $E_B\subseteq E$, such that any subgraph $\hat{G} = (V,\hat{E})$ containing $E_B$, i.e., $E_B\subseteq \hat{E}\subseteq E$ satisfies
\begin{equation}
\gamma(\hat{G}, V_\ell) \geq \gamma(G, V_\ell).
\end{equation}
\end{definition} 

In other words, any subgraph $\hat{G} = (V,\hat{E})$ of $G$ containing backbone edges $E_B$, has at least the same controllability as $G$. Thus, preserving backbone edges guarantees that controllability does not deteriorate in a subgraph $\hat{G}$. A backbone with the minimum edge set is referred to as the \emph{minimum backbone graph} $B^* = (V, E^*)$. We aim to compute $B^*$.

\vspace{2 mm}
\noindent\fbox{\begin{varwidth}{\dimexpr\linewidth-1\fboxsep-2\fboxrule\relax}
    \begin{prob}
    \label{prob:1}
        Given a graph $G = (V, E)$ and a leader set $V_\ell$, find the minimum controllability backbone graph.
    \end{prob}
\end{varwidth}}
\vspace{1 mm}

Figure~\ref{fig:backbone_illustration} illustrates the idea of a controllability backbone. For a given $G$ and $V_\ell = \{v_4,v_6,v_7\}$, the dimension of SSC is $\gamma(G,V_\ell)=8$. A minimum backbone $B^*$ is shown in Figure~\ref{fig:backbone_illustration}(b). Any subgraph $\hat{G}$ (of $G$) containing $B^*$ also has $\gamma(\hat{G},V_\ell) = 8$.

\begin{figure}[htb]
    \centering
    \begin{subfigure}[b]{0.15\textwidth}
	\centering
        \includegraphics[scale=0.41]{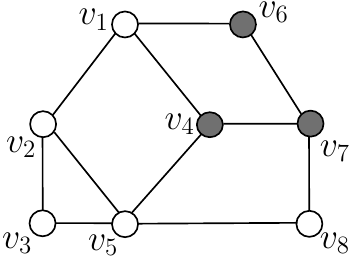}
	\caption{$G$}
	\end{subfigure}
	\begin{subfigure}[b]{0.15\textwidth}
	\centering
         \includegraphics[scale=0.41]{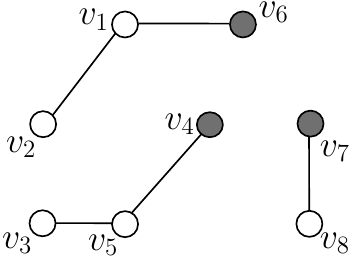}
	\caption{$B^*$}
	\end{subfigure}
	\begin{subfigure}[b]{0.15\textwidth}
	\centering
	\includegraphics[scale=0.41]{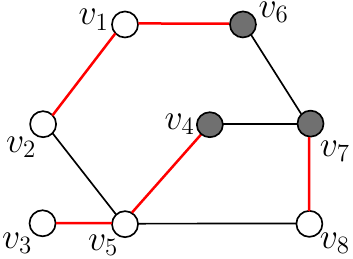}
	\caption{$\hat{G}$}
	\end{subfigure}
    \caption{(a) A graph $G$. (b) A minimum controllability backbone $B^*$ of $G$. (c) A subgraph $\hat{G}$ of $G$ containing the backbone (red edges).}
    \label{fig:backbone_illustration}
\end{figure}

The minimum backbone problem relies on the computation $\gamma(G,V_\ell)$ (as in \eqref{eq:min_rank}), which is a computationally arduous task. To address this challenge, it is common to compute tight lower bounds on $\gamma(G,V_\ell)$ instead of $\gamma(G,V_\ell)$ when dealing with SSC-related problems.
As a result, we also modify the controllability backbone problem and focus on obtaining a sparse subgraph of a given $G$ whose existence within any subgraph $\hat{G}\subseteq G$ guarantees that $\hat{G}$ has the same or greater value of lower bounds on the dimension of SSC as $G$. To accomplish this, we consider two widely used lower bounds, including (1) a zero forcing set-based bound and (2) a bound based on the distances between vertices. In the forthcoming sections, we will elaborate on these bounds and their application to the controllability backbone problem.

\section{Zero forcing for Controllability backbone}
\label{sec:zfs}
Zero forcing is a rule-based coloring of vertices in a graph. The main idea is to initiate the coloring process with a small subset of initially colored vertices which eventually color other vertices based on some rules. Zero forcing has several network applications and provides a tight lower bound on the network's SSC, as we explain below~\cite{monshizadeh2014zero}.

\subsection{Zero Forcing-based Lower Bound on SSC}
First, we define the zero forcing process and related terms and then explain the SSC bound based on the zero forcing phenomenon.
\label{subsec:ZFS}
\begin{definition} (\emph{Zero forcing (ZF) Process}) Consider a graph $G=(V,E)$, such that each $v\in V$ is colored either BLACK or WHITE initially. The ZF process is to iteratively change the color of WHITE vertices to BLACK using the following rule until no further color changes are possible.

\emph{Color change rule: If $v\in V$ is colored BLACK and has exactly one WHITE neighbor $u$, change the color of $u$ to BLACK.}
\end{definition}

We say that $v$ \emph{infected} $u$ if the color of WHITE vertex $u$ is changed to BLACK by some BLACK vertex $v$. 

\begin{definition}
\label{def:derived}(\emph{Derived Set})
Consider a graph $G = (V,E)$ with $V_\ell\subseteq V$ as the set of initial BLACK vertices. Then, the set of BLACK vertices obtained at the end of the ZF process is the derived set~\cite{work2008zero}, denoted by $dset(G, V_\ell)$, and $|dset(G, V_\ell)| = \zeta(G, V_\ell)$. When the context is clear, we will drop the parameter $V_\ell$.
\end{definition}

The set of initial BLACK vertices $V_\ell$ is also referred to as the \emph{input or leader set}. For a given $V_\ell$, $dset(G, V_\ell)$ is unique \cite{work2008zero}. Now, we define the zero forcing set.

\begin{definition} (\emph{Zero Forcing Set (ZFS)}) For a graph $G=(V,E)$, $V_\ell\subseteq V$ is a ZFS if and only if $dset(G,V_\ell) = V$. We denote a ZFS of $G$ by $Z(G)$. 
\end{definition}
Figure~\ref{fig:ZFS} illustrates zero forcing through a set of input vertices and the corresponding derived set. 

\begin{figure}[ht]
	\centering
	\includegraphics[scale=0.425]{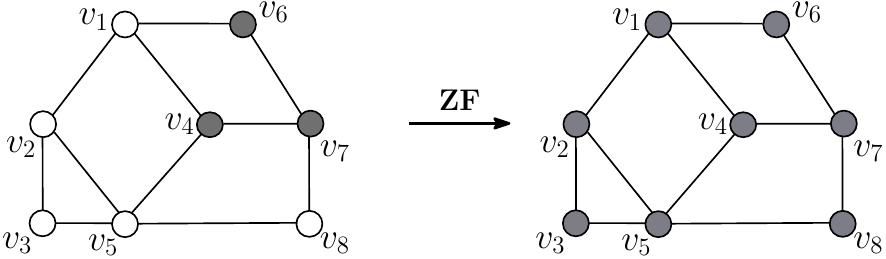}
    \caption{$V_\ell = \{v_4,v_6,v_7\}$ is the 
    input set. After the ZF process, $dset(G, V_\ell) =V$, as indicated by the black vertices. Hence, $V_\ell$ is a ZFS.}
    \label{fig:ZFS}
\end{figure}

The zero forcing phenomenon is significant in characterizing the network's SSC~\cite{monshizadeh2014zero,trefois2015zero,yaziciouglu2022strong}. In particular, the size of the derived set for a given set of input vertices provides a lower bound on the dimension of SSC.

\begin{theorem}
\cite{monshizadeh2015strong} For any network $G = (V, E)$ with the leaders $V_\ell \subseteq V $,
$$
\zeta(G, V_\ell) \leq \gamma(G, V_\ell),
$$    
where $\zeta(G,V_\ell)$ is the size of the derived set with $V_\ell$ as input vertices, and $\gamma(G,V_\ell)$ is the dimension of SSC (as in \eqref{eq:min_rank}).
\end{theorem}
\begin{proof} Proof follows from Lemma 4.2 in \cite{monshizadeh2015strong}, which shows that for a set of state matrices, the controllable subspace always
contains a $|dset(G, V_\ell)|$-dimensional subspace.
\end{proof}

\subsection{ZFS-based Backbone}
We are interested in finding a controllability backbone that will maintain the zero forcing bound $\zeta(G,V_\ell)$ for a given leader set $V_\ell$. The idea is to identify a subset of edges $E_{B_Z}$ in a given $G = (V,E)$ with a leader set $V_\ell$ such that the ZFS-based controllability bound in \emph{any} subgraph of $\hat{G} = (V,\hat{E})$ containing those edges, (i.e., $E_{B_Z}\subseteq \hat{E}$) is preserved. We formally define the ZFS-based backbone as follows:

\begin{definition} (\emph{ZFS-based Backbone}) Given a graph $G = (V, E)$ and a leader set $V_\ell$, the ZFS-based backbone is a subgraph $B_z = (V, E_{B_z})$, such that any subgraph $\hat{G} = (V,\hat{E})$, where $E_{B_Z}\subseteq \hat{E}\subseteq E$ satisfies the following:
 \begin{equation*}
            \zeta(\hat{G}, V_\ell) \ge \zeta(G, V_\ell).
        \end{equation*}    
\end{definition}

Thus, the dimension of SSC in any subgraph of $G$ containing the ZFS-based backbone is at least $\zeta(G,V_\ell)$, or in other words, $\gamma(\hat{G},V_\ell)\ge \zeta(G,V_\ell)$. Our goal is to find the ZFS-based backbone with the minimum number of edges. 

\vspace{2 mm}
\noindent\fbox{\begin{varwidth}{\dimexpr\linewidth-1\fboxsep-2\fboxrule\relax}
    \begin{prob}
    \label{Prob:ZFS}
        Given a graph $G = (V, E)$ and a leader set $V_\ell$, find a minimum ZFS-based backbone.
    \end{prob}
\end{varwidth}}
\vspace{2 mm}

In \cite{monshizadeh2014zero}, authors show that a leader set $V_\ell$ renders the network strong structurally controllable if and only if $V_\ell$ is a zero forcing set (ZFS) of the network graph $G$. Thus, the ZFS-based backbone is essentially the controllability backbone (as defined in Section~\ref{sec:Ctrb_backbone}) if $V_\ell$ is a ZFS of $G$. Algorithm~\ref{alg:backbone_zfs} solves Problem~\ref{Prob:ZFS} and computes a minimum ZFS-based backbone. The main idea is to run the ZF process and iteratively select an edge through which some BLACK vertex colors its WHITE neighbor (thus, increasing the size of the derived set). 

\begin{algorithm}[ht]
\caption{Computing ZF-based controllability backbone}
\label{alg:backbone_zfs}
    \begin{algorithmic}[1]
    \renewcommand{\algorithmicrequire}{\textbf{Input:}}
    \renewcommand{\algorithmicensure}{\textbf{Output:}}
    \Require $G$, $V_\ell$
    \Ensure ZFS-based backbone $B_z = (V, E_{B_z})$\\
        \textbf{Initialize:} $E_{B_z} \gets \emptyset$,\quad $dset(G, V_\ell) \gets V_\ell$ \texttt{(The  set of initial BLACK vertices)} 
        \While {there exits a BLACK vertex $v$ with exactly one WHITE neighbor $u$}\\
            \;\;\;\;\;$dset(G, V_\ell) \gets dset(G, V_\ell) \cup \{u\}$\\
            \;\;\;\;\;$E_{B_z} \gets E_{B_z} \cup \{(u, v)\}$
        \EndWhile
    \end{algorithmic}
\end{algorithm}

\begin{theorem}
\label{thm:zfs_backbone}
Consider a graph $G=(V,E)$ and a leader set $V_\ell\subseteq V$, where $|V| = n$ and $|V_\ell| = m$. Algorithm~\ref{alg:backbone_zfs} returns a minimum ZFS-based backbone $B_{z} = (V,E_{B_z})$ in $O(n^2)$ time. Moreover, $|E_{B_z}| = \zeta(G,V_\ell) - m$. 
\end{theorem}

\begin{proof}
First, we show that the size of the graph returned by Algorithm~\ref{alg:backbone_zfs}, i.e., $|E_{B_z}|$, is $ |dset(G, V_\ell)| - m$. We start with an empty graph and $dset(G, V_\ell)$ contains only the leader vertices. Every time we add an edge to the output graph, we include a vertex to $dset(G, V_\ell)$. Therefore, the number of edges in the graph is the size of the final $dset(G, V_\ell)$ minus the number of leaders.\\
We need to show that the graph returned is indeed a ZFS-based backbone graph. We prove this by showing that for every graph $\hat{G}$ with $B_z\subseteq \hat{G} \subseteq G$, the size of the derived set $\zeta(\hat{G}, V_\ell)$ is more than or equal to $\zeta(G, V_\ell)$, i.e., the size of the derived set of the original graph $G$ for a given leader set $V_\ell$. We propose to copy the zero forcing process as it is done on the graph $G$. At an arbitrary step of this process, a BLACK colored vertex $v$ forces the color of a WHITE color neighbor $u$ to BLACK. The edge $(u,v)$ is preserved in the graph $B_z$ and in every super graph $\hat{G}$ of $B_z$, therefore, $u$ must be a WHITE neighbor of a BLACK colored vertex $v$ in graph $\hat{G}$ at this step of zero forcing process. To complete a valid step, though, we need $u$ to be the only WHITE neighbor of $v$. At this step in the original graph $G$, all vertices in $\mathcal{N}_G(v)\setminus {u}$ are colored BLACK. As $\hat{G}\subseteq G$, we also have $(\mathcal{N}_{\hat{G}}(v)\setminus {u}) \subseteq (\mathcal{N}_G(v)\setminus {u})$. Thus all neighbors of $v$ in $\hat{G}$ except $u$ must be colored BLACK at this step. Therefore, the zero forcing process can be completed for every graph $\hat{G}$ with $B_z\subseteq \hat{G} \subset G$ and the graph returned by the algorithm is indeed a ZFS-based backbone of $G$, i.e., $\zeta(\hat{G}, V_\ell) \geq \zeta(G, V_\ell)$.
In particular, $|dset(G, V_\ell)| = \zeta(B_z,V_\ell) = \zeta(G,V_\ell)$. We have the size of the backbone graph as in the statement, $|E_{B_z}| = |dset(G, V_\ell)| - m = \zeta(G, V_\ell) - m$.\\
Regarding the time complexity of the algorithm, we only run the zero forcing process once, therefore, the time complexity of Algorithm~\ref{alg:backbone_zfs} is bounded by the time complexity of computing the derived set of a graph with a given leader set, which is $O(n^2)$.
This concludes the proof.
\end{proof}

\begin{figure}[ht]
    \centering
     \begin{subfigure}{0.49\linewidth}
         \includegraphics[width=0.8\linewidth]{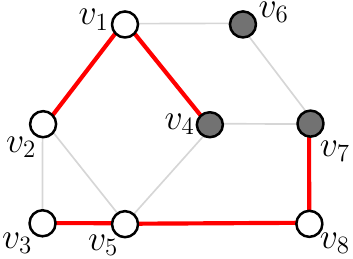}
     \end{subfigure}
     \begin{subfigure}{0.49\linewidth}
         \includegraphics[width=0.8\linewidth]{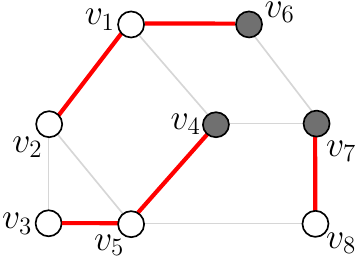}
     \end{subfigure}
    \caption{Two distinct ZFS-based backbones of $G$ in Figure \ref{fig:ZFS}.}    
    \label{fig:backbone_zfs}
\end{figure}

Backbone obtained by the ZFS method is not necessarily unique. For instance, Figure \ref{fig:backbone_zfs} illustrates two distinct ZFS-based backbones of $G$ (in Figure \ref{fig:ZFS}) and $V_\ell = \{v_4,v_6,v_7\}$. Though there can be multiple distinct backbones, they all have the same number of edges (as in Theorem~\ref{thm:zfs_backbone}). We characterize the edges in any ZFS-backbone into two categories,  \emph{necessary} and \emph{contingent} edges. The \emph{necessary} edges are the ones that must be included in \emph{every} ZFS-backbone $B_z$ of $G$ for a given $V_\ell$, whereas all the other edges of the backbone are the \emph{contingent} edges. In the example of Figure~\ref{fig:backbone_zfs}, edge $(v_1, v_2)$ is the necessary edge. For a given $G$ and $V_\ell$, a simple characterization of necessary edges is that if they are removed from $G$, then the size of the derived set is reduced. In other words, if $e$ is a necessary edge and $G\setminus e$ denotes a graph obtained from $G$ by removing the edge $e$, then $dset(G,V_\ell) > dset (G\setminus e, V_\ell)$. 

Additionally, we note that when $V_\ell$ is a ZFS of $G$ (i.e., $dset(G,V_\ell) = V$), then $\zeta(G,V_\ell) = \gamma(G,V_\ell) = |V|$, implying that a minimum ZFS-based backbone is also a minimum controllability backbone (as in Problem~\ref{prob:1}). However, when the leader set $V_\ell$ is not a ZFS, then the distance-based bound on the dimension of SSC is typically better than the ZFS-based bound~\cite{yaziciouglu2022strong}. Next, we discuss the distance-based bound and apply it to the controllability backbone problem.

\section{Graph Distances for Controllability backbone}
\label{sec:dist}

In this section, we design a controllability backbone using a bound on the network SSC based on the graph distances between vertices in the underlying network graph. First, we introduce the distance-based bound on the dimension of SSC~\cite{yaziciouglu2016graph}. We then frame the notion of the distance-based controllability backbone and provide an algorithm to compute such a backbone. In Section~\ref{sec:comp}, we compare the ZFS-based and distance-based backbones.

\subsection{Distance-based Lower Bound}
Assuming $m$ leaders $V_\ell = \{\ell_1, \ell_2, \cdots , \ell_m\}$ in a leader-follower network $G = (V,E)$, we define the \emph{distance-to-leader (DL) vector} for each $v_i \in V$ as
$$
D_i = \left[\begin{array}{lllll}
    d(\ell_1, v_i) & d(\ell_2, v_i) & \cdots & d(\ell_{m}, v_i)  \\
\end{array}\right]^T \in \mathbb{Z}^m.
$$

The $j^{th}$ component of $D_i$, denoted by $[D_i]_j$, is $d(\ell_j, v_i)$, i.e., the distance between leader $\ell_j$ and vertex $v_i$. Figure \ref{fig:PMI_example} shows DL vectors of vertices in a graph $G$ with leaders $V_\ell = \{v_4, v_6\}$. Next, we define a \emph{sequence} of distance-to-leader vectors, called \emph{pseudo-monotonically increasing sequence}~\cite{yaziciouglu2016graph}.

\begin{definition} (\emph{Pseudo-monotonically Increasing Sequence (PMI))} A sequence $\calD = [\calD_1 \; \; \calD_2 \;\; \cdots \;\calD_k]$ of distance-to-leader vectors is a PMI if for any vector $\calD_i$ in the sequence, there is some coordinate $\pi(i)\in \{1, 2,\cdots, m\}$ such that
\begin{equation}
\label{eq:PMI}
    [\calD_i]_{\pi(i)} < [\calD_j]_{\pi(i)}, \; \forall j > i. 
    \end{equation}
\end{definition} 
We say $[\calD_i]_{\pi(i)}$ satisfies the PMI property at coordinate $\pi(i)$.

The PMI property \eqref{eq:PMI} essentially gurantees that for each vector $\mathcal{D}_i$ in the PMI sequence, there is some index/coordinate $\pi(i)$ such that the values of all the subsequent vectors at the coordinate $\pi(i)$ are strictly greater than $[\mathcal{D}_i]_{\pi(i)}$. An example of PMI sequence of six vectors is shown in \eqref{eqn:PMI_seq}, where the coordinates of circled values are the ones where the PMI property is satisfied. 

Next, we note that each vector in a PMI sequence is a DL vector of some vertex in the graph. However, multiple vertices can have the same DL vectors. For example, the DL vectors of $v_2$ and $v_8$ are the same. Thus, to explicitly specify the vertex whose DL vector appears in the PMI sequence, we introduce the \emph{distance-to-leader} mapping.

\begin{definition}
\label{def:dlm}
(\emph{Distance-to-Leader Mapping (DLM))} Let $\mathcal{D}$ be a PMI sequence. For each $\calD_i\in\calD$, a Distance-to-Leader Mapping (DLM), denoted by $f(\calD_i)$, is a vertex whose DL vector is $\calD_i$, i.e., $\calD_i = D_{f(\calD_i)}$. 
\end{definition}

To further clarify, note the following notations:

$\mathcal{D}_i$: \hspace{0.5cm} $i^{th}$ vector in the PMI sequence,

$D_v$: \hspace{0.5cm} DL vector of vertex $v$.

Figure~\ref{fig:PMI_example} illustrates these ideas. For the graph $G$ in Figure~\ref{fig:PMI_example} and $V_\ell = \{v_4, v_6\}$, a PMI sequence of length six can be constructed as
\begin{equation}
\label{eqn:PMI_seq}
    \mathcal{D} = \left[
        \begin{bmatrix}
            \encircle{0}\\2 \end{bmatrix},
        \begin{bmatrix}
            2\\\encircle{0} \end{bmatrix}, 
        \begin{bmatrix}
            1\\\encircle{1} \end{bmatrix},
        \begin{bmatrix}
            \encircle{1}\\3 \end{bmatrix},
        \begin{bmatrix}
            2\\\encircle{2} \end{bmatrix},
        \begin{bmatrix}
            \encircle{2}\\3 \end{bmatrix}
    \right].
\end{equation}

\begin{figure}[ht]
    \centering
    \includegraphics[scale = 0.55]{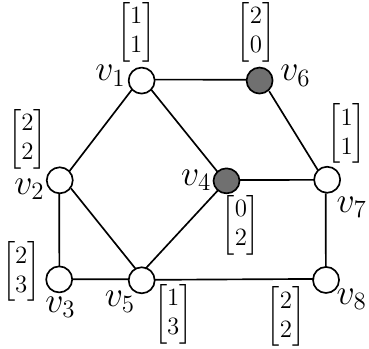}
    \caption{A network with two leaders $V_\ell = \{\ell_1, \ell_2\} = \{v_4, v_6\}$, along with the DL vectors of vertices. A PMI sequence of length six is $\calD = [\calD_1\; \calD_2\; \cdots\; D_6] = [D_{v_4}\; D_{v_6}\; D_{v_7}\; D_{v_5}\; D_{v_8}\; D_{v_3}]$.}
    \label{fig:PMI_example}
\end{figure}

A PMI sequence of DL vectors is related to the network SSC. In fact, the length of PMI sequence provides a tight lower bound on the dimension of SSC $\gamma(G,V_\ell)$, as stated in the following result.

\begin{theorem}~\cite{yaziciouglu2016graph}
    If $\delta(G, V_\ell)$ or simply $\delta(G)$, is the length of the longest PMI sequence of distance-to-leader vectors in a network $G = (V, E)$ with $V_\ell$ leaders, then 
    \begin{equation}
        \label{eq:distance_bound}
    \delta(G, V_\ell) \leq \gamma(G, V_\ell),
    \end{equation}
    where $\gamma(G,V_\ell)$ is the dimension of SSC (as in \eqref{eq:min_rank}).

\end{theorem}

\subsection{Distance-based Backbone}
Here, we will use the distance-based bound to formulate a controllability backbone problem. Then, we will provide and analyze an algorithm for computing such a backbone.

\begin{definition} (\emph{Distance-based Backbone}) Given a graph $G = (V, E)$ and a leader set $V_\ell$, the distance-based backbone is a subgraph $B_d = (V, E_{B_d})$ such that any subgraph $\hat{G} = (V,\hat{E})$, where $E_{B_d}\subseteq \hat{E}\subseteq E$ satisfies the following:
 \begin{equation*}
            \delta(\hat{G}, V_\ell) \geq \delta(G, V_\ell).
        \end{equation*}    
\end{definition}
It basically means that any subgraph $\hat{G}$ containing backbone edges $E_{B_d}$ has the longest PMI sequence with $V_\ell$ leaders of length more or equal to the longest PMI sequence as in $G$, and thus, has at least the same controllability bound as in $G$. As a result, any $\hat{G}$ containing the backbone edges satisfies $\delta(G,V_\ell) \le \gamma(\hat{G},V_\ell)$. So, our goal is to find the minimum distance-based backbone.

\vspace{2 mm}
\noindent\fbox{\begin{varwidth}{\dimexpr\linewidth-1\fboxsep-2\fboxrule\relax}
    \begin{prob}
        Given a graph $G = (V, E)$ and a leader set $V_\ell$, find a minimum distance-based backbone.
    \end{prob}
\end{varwidth}}
\vspace{2 mm}

Algorithm~\ref{alg:backbone_distance} presents a scheme to compute a minimal distance-based backbone of a given $G=(V, E)$ and leader set $V_\ell = \{\ell_1,\ell_2,\cdots,\ell_m\}$. The input to the algorithm is a PMI sequence $\mathcal{D} = \left[\begin{array}{cccccc}\mathcal{D}_1 & \mathcal{D}_2&\cdots& \mathcal{D}_{\delta}\end{array}\right]$ of length $\delta(G,V_\ell) = \delta$. Additionally, we also know the corresponding distance-to-leader mapping $f(\mathcal{D}_i)$ for each $\mathcal{D}_i\in \mathcal{D}$ (i.e., vertices whose DL vectors appear in $\mathcal{D}$). We note that PMI sequence and DLM can be easily computed using methods in \cite{shabbir2022computation}. The main idea of Algorithm~\ref{alg:backbone_distance}  is to maintain the edges such that for each $\mathcal{D}_i\in\mathcal{D}$, the distance between $f(\mathcal{D}_i)$ and the leader $\ell_{\pi(i)}\in V_\ell$ is preserved. Recall that $\pi(i)$ is the coordinate of $\mathcal{D}_i$ at which the PMI property is satisfied. The details are outlined below.

\begin{algorithm}[ht]
\caption{Computing distance-based backbone}
\label{alg:backbone_distance}
    \begin{algorithmic}[1]
    \renewcommand{\algorithmicrequire}{\textbf{Input:}}
    \renewcommand{\algorithmicensure}{\textbf{Output:}}
    \Require $G$, $V_\ell$, PMI sequence $\calD = [\calD_1 \;\; \calD_2 \;\cdots \; \calD_{\delta(G,V_\ell)}]$
    \Ensure Distance-based backbone $B_d = (V, E_{B_d})$
        \For {$i=1$ to $\delta(G, V_\ell)$}\\
        \;\;\; compute $f(\calD_i)$, \hspace{0.1in}\texttt{\% DLM of $\mathcal{D}_i$.}\\
        \;\;\;\;compute $d(\ell_{\pi(i)},f(\calD_i))$ \texttt{ \% distance between leader $\ell_{\pi(i)}$ and vertex $f(\calD_i)$.} \\
        \;\;\; $\calP_i \gets$ edges on the shortest path between $\ell_{\pi(i)}$ and $f(\calD_i)$. \hspace{0.1in}\texttt{\% If there are multiple, choose any shortest path.}
        \EndFor\\
        $ E_{B_d} \gets$  ${\bigcup}_i$  $\calP_i$
    \end{algorithmic}
\end{algorithm}

\emph{Example:} We illustrate Algorithm~\ref{alg:backbone_distance} using the example in Figure~\ref{fig:PMI_example} with the leader set $V_\ell = \{\ell_1,\ell_2\} = \{v_4,v_6\}$. A PMI of length $\delta(G,V_\ell)=6$ is given in \eqref{eqn:PMI_seq}. The corresponding DLMs are: $f(\mathcal{D}_1) = v_4$, $f(\mathcal{D}_2) = v_6$, $f(\mathcal{D}_3) = v_7$, $f(\mathcal{D}_4) = v_5$, $f(\mathcal{D}_5) = v_8$, and $f(\mathcal{D}_6) = v_3$. For $i=1$ (first iteration), we have $\mathcal{D}_1 = \scriptsize{\left[\begin{array}{cc}0\\2\end{array}\right]}$, $\pi(1) = 1$. So, edges in the shortest path between $f(\mathcal{D}_1)$ and $\ell_1$ needs to be preserved. Since  $f(\mathcal{D}_1) = v_4 = \ell_1$, $\mathcal{P}_1 = \emptyset$. Similarly, in the second iteration ($i = 2$), we get $\mathcal{P}_2 = \emptyset$. For $i = 3$ (third iteration), we have $\mathcal{D}_3 = \scriptsize{\left[\begin{array}{cc}1\\1\end{array}\right]}$, and $\pi(3) = 2$. So, we need to preserve edges that appear in the shortest path between $f(\mathcal{D}_3) = v_7$ and $\ell_{\pi(3)} = \ell_{2} = v_6$. Since $v_6$ and $v_7$ are adjacent, $\mathcal{P}_3 = \{(v_6,v_7)\}$. Continuing this way, we will get the set of edges $E_{B_d}$ in the distance-based backbone at the end of $\delta(G, V_\ell)=6$ iterations. Figure~\ref{fig:dis_highlight} illustrates the distance-based backbone.

\begin{figure}[ht]
    \centering
\includegraphics[scale=0.55]{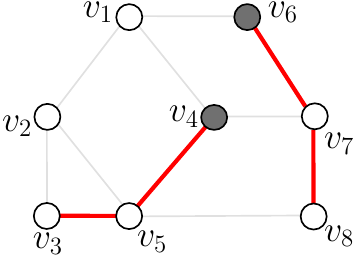}
    \caption{Distance-based backbone $B_d = (V,E_{B_d})$. Edges in $B_d$ are highlighted red.}
    \label{fig:dis_highlight}
\end{figure}
Next, we show that the above algorithm always returns a distance-based backbone.
\begin{theorem}
    \label{thm:distance}
    Consider a graph $G=(V,E)$ and a leader set $V_\ell\subseteq V$, where $|V| = n$ and $|V_\ell| = m$. Algorithm~\ref{alg:backbone_distance} returns a distance-based backbone $B_{d} = (V,E_{B_d})$ in $O(n^3)$ time.
\end{theorem}
\begin{proof} 
We show that the backbone graph $B_d$ will have a PMI of the same maximum length as the PMI of the given graph $G$. Additionally, any subgraph $\hat{G}$ of $G$, where $B_d \subseteq \hat{G} \subseteq G$,
will have the longest length of the PMI sequence of size more or equal to the longest length PMI sequence of $G$. Let's assume $\calD$ is the PMI of the original graph $G$ and the DLM $f(\calD_i)$ for each of the DL vectors $\calD_i$ in $\calD$ are known. 
For any $1\leq i\leq \delta(G, V_\ell)-1$, let us consider the two consecutive DL vectors $\calD_i$ and $\calD_{i+1}$ for the given PMI $\calD$. To keep the distance vector $\calD_i$ in the PMI sequence, the constraint \eqref{eq:PMI} must be satisfied at $[\calD_i]_{\pi(i)}$. 
For $[\calD_i]_{\pi(i)}$, according to Algorithm \ref{alg:backbone_distance}, we maintain $d(f(\calD_i),\ell_{\pi(i)})$, the distance between vertex $f(\calD_i)$ and the leader $\ell_{\pi(i)}$ i.e. we preserve all the edges in the shortest path $\calP_i$ in the backbone graph $B_d$. Similarly, for $[D_{i+1}]_{\pi(i+1)}$, we will be maintaining $d(f(\calD_{i+1}),\ell_{\pi(i+1)})$, the distance between vertex $f(\calD_{i+1})$ and the leader $\ell_{\pi(i+1)}$ by preserving all the edges in $\calP_{i+1}$. While maintaining the respective distances for vertices $f(\calD_{i})$ and $f(\calD_{i+1})$, we can remove the rest of the edges that do not necessarily maintain these paths $\calP_i$ and $\calP_{i+1}$. The PMI constraint in \eqref{eq:PMI} states that the distance for vector $[\calD_{i+1}]_{\pi(i)}$ must strictly be greater than the distance between $f(\calD_{i})$ and $\ell_{i}$. As \emph{removing any edge can not decrease the distance between any two vertices}, the constraint \eqref{eq:PMI} is always satisfied for $\pi(i)$ between $\calD_i$ and $\calD_{i+1}$. The PMI constraint between $[D_{i}]_{\pi(i)}$, and $[D_{i+1}]_{\pi(i)}$ will not be violated $\forall i = 1,\cdots, \delta(G, V_\ell)-1$. Hence, as long as the paths $\calP_i$ and $\calP_{i+1}$ are preserved, $\delta(\hat{G}, V_\ell)$ will be preserved for any $\hat{G} \subseteq G$.  Thus, Algorithm \ref{alg:backbone_distance} will return a backbone $B_d$.
\newline
Its worst-case time complexity is $O(n \times|E|) = O(n^3)$ where we iterate over all the vertices and can use breadth-first search ($O(|E|)$) to find the shortest path for each vertex $f(\calD_i)$ in the PMI sequence.
\end{proof}

We note that for a given $G$ and leader set $V_\ell$, multiple distinct PMIs of the same length can exist. Moreover, there can be multiple distance-to-leader mappings for a given PMI $\mathcal{D}$. At the same time, for a given $\mathcal{D}$, there can be multiple ways to assign a coordinate $\pi(i)$ to $\mathcal{D}_i\in\mathcal{D}$ while satisfying the PMI property. Thus, for a given $G$ and $V_\ell$, multiple PMIs can have the same length $\delta(G, V_\ell)$ but distinct distance-to-leader mappings or $\pi(i)$ combinations. The distance-based backbone returned by Algorihm~\ref{alg:backbone_distance} depends on the above-mentioned factors, i.e., the input PMI sequence, $\pi(i)$ combinations in the PMI sequence, and the corresponding distance-to-leader mappings. Hence, the distance-based backbone returned by Algorihm~\ref{alg:backbone_distance} may be distinct for different PMI sequences with the same lentgh $\delta(G, V_\ell)$. We illustrate this in the below example, and then provide upper and lower bounds on the number of edges in the distance-based backbone computed by Algorithm~\ref{alg:backbone_distance}.

\emph{Example:} Consider the graph in Figure~\ref{fig:dist_bckbn}(a) with two leaders $V_\ell = \{v_1,v_2\}$ and the corresponding DL vectors of all the vertices. We can obtain two distinct PMI sequences, say $\tilde{\mathcal{D}}$ and $\breve{\mathcal{D}}$, each of length $\delta(G, V_\ell) = 4$. These sequences along with the coordinates satisfying the PMI property ($\pi(i)$) and the corresponding DL mappings are given below:

\begin{equation*}
\label{eqn:PMI_seq_1}
    \tilde{\mathcal{D}}= [D_{v_1}\;D_{v_2}\;D_{v_5}\;D_{v_3}]\\
    =\left[
        \begin{bmatrix}
            \encircle{0}\\1 \end{bmatrix},
        \begin{bmatrix}
            1\\\encircle{0} \end{bmatrix}, 
        \begin{bmatrix}
            \encircle{1}\\1 \end{bmatrix},
        \begin{bmatrix}
            2\\\encircle{1} \end{bmatrix}
    \right],
\end{equation*}

\begin{equation*}
\label{eqn:PMI_seq_2}
    \breve{\mathcal{D}} = [D_{v_1}\;D_{v_2}\;D_{v_4}\;D_{v_3}]
    =\left[
        \begin{bmatrix}
            \encircle{0}\\1 \end{bmatrix},
        \begin{bmatrix}
            1\\\encircle{0} \end{bmatrix}, 
        \begin{bmatrix}
            \encircle{1}\\2 \end{bmatrix},
        \begin{bmatrix}
            \encircle{2}\\1 \end{bmatrix}
    \right].
\end{equation*}

Though $\tilde{\mathcal{D}}$ and $\breve{\mathcal{D}}$ are of the same length, they are different in terms of the DL vectors they contain and the corresponding DL mappings. The distance-based backbones obtained by $\tilde{\mathcal{D}}$ and $\breve{\mathcal{D}}$ as a result of Algorithm~\ref{alg:backbone_distance} are illustrated in Figures~\ref{fig:dist_bckbn}(b) and (c), respectively. We note that the number of edges is different in the two backbones. We obtain a minimum distance-based backbone due to $\tilde{\mathcal{D}}$ consisting of $\delta(G, V_\ell) - |V_\ell| = 2$ edges.

\begin{figure}[ht]
    \centering
     \begin{subfigure}{0.3\linewidth}
     \centering
    \includegraphics[scale=0.5]{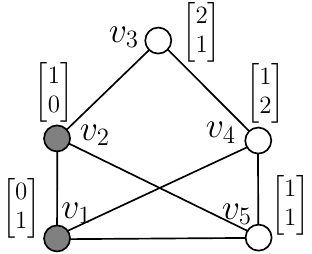}
         \caption{}
     \end{subfigure}
     \hfill
     \begin{subfigure}{0.3\linewidth}
     \centering
         \includegraphics[scale=0.5]{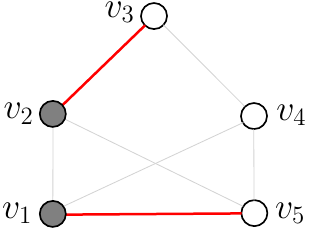}
         \caption{}
     \end{subfigure}
        \hfill
     \begin{subfigure}{0.3\linewidth}
     \centering
         \includegraphics[scale=0.5]{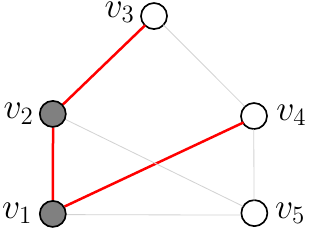}
         \caption{}
     \end{subfigure}
    \caption{(a) A graph $G$ with two leaders. (b) Minimum distance-based backbone due to $\tilde{\mathcal{D}}$. (c) Distance-based backbone as a result of $\breve{\mathcal{D}}$ as an input to Algorithm~\ref{alg:backbone_distance}.}
    \label{fig:dist_bckbn}
\end{figure}

Next, we state the lower and upper bounds on the number of edges in the distance backbone returned by Algorithm~\ref{alg:backbone_distance}.

\begin{prop}
\label{thm:distance_lowerbound}
For a given graph $G$ and a leader set $V_\ell$, the lower bound on the number of edges in the distance-based backbone $B_d$ computed using Algorithm \ref{alg:backbone_distance} is $\delta(G, V_\ell) - m$, where $m$ is the number of leaders. 
\end{prop}

\begin{proof}
    Let $\mathcal{D} = [\calD_1\;\calD_2\;\cdots\;\calD_{\delta(G, V_\ell)}]$ be an input PMI sequence in Algorithm~\ref{alg:backbone_distance}. If a vertex $f(\calD_i)$ is a part of the PMI sequence $\calD$, it must have a path to at least one leader $\ell_\pi(i)$. Let the number of edges in $B_d$ is less than $\delta(G, V_\ell) - m$, then there must be one or more vertex that does not have a path to any leader vertex and can not be a part of the longest PMI sequence. So, the number of edges can not be less than the number of distance vectors of the follower vertices in the PMI sequence i.e. $\delta(G, V_\ell) - m$. 
\end{proof}

Figure~\ref{fig:dist_bckbn}(b) illustrates an example of a minimum distance backbone. An upper bound on the number of edges in the distance backbone due to Algorithm~\ref{alg:backbone_distance} is given below.

\begin{prop}
\label{thm:distance_upperbound}
For a given graph $G$, a leader set $V_\ell$, and a maximum length PMI $\calD$, the upper bound on the number of edges in $B_d$ computed using Algorithm \ref{alg:backbone_distance} is $\dbinom{m}{2} + \dbinom{\delta(G, V_\ell) - m+1}{2}$.
\end{prop}

\begin{proof}
    Let $S = \{D_1, D_2\cdots,D_n\}$ be the set of all DL vectors for a given leader-follower network and let $C_p \subseteq S$ be the set of all DL vectors that can be assigned as the $p^{th}$ element of a PMI sequence $\calD$. Once a vector from $C_p$ is assigned as the $p^{th}$ element of the sequence, $D_p$, and an index $\pi(p)$ satisfying \eqref{eq:PMI} is chosen, the resulting $C_{p+1}$ can be obtained from $C_p$ as
    \begin{equation}
    \label{eq:C_set}
        C_{p+1} = \{ D_i \in C_p | [D_i]_{\pi(i)} > [\calD_p]_{\pi(i)}\}
    \end{equation}
    The \emph{Proposition 4.2} of \cite{yaziciouglu2016graph} states that \eqref{eq:C_set} can be used to find the longest length PMI sequence using Algorithm $I$ of \cite{yaziciouglu2016graph}. We know that $C_p$ is the multiset, where each element $C_{p,i}$ is the $C_p$ resulting from \eqref{eq:C_set} for specific choices of $\pi(1), \cdots, \pi(p-1)$ subject to the corresponding PMI sequences satisfying 
    \begin{equation*}
        [\calD_{p}]_{\pi(p)} = \min_{D_i \in C_p} [D_{i}]_{\pi(p)}.
    \end{equation*}
    We know that there must be a potential candidate $\calD_p$ for each $p \leq \delta(G, V_\ell)$ until $C_p = \phi$. This shows that each distance vector $\calD_i$ can not have a distance more than $i$ at $\pi(i)$ which can result in at most $\dbinom{\delta(G, V_\ell)}{2}$ unique edges in the backbone. 
    \newline
    We know that the distance vectors of the leader vertices will always be a part of the longest PMI sequence since the distance of each leader to itself (zero) satisfies the PMI rule. Hence, we essentially don't need to have any edges for all the distance vectors belonging to the leader vertices and we only have positive distances for the follower vertices. Hence, we will only need $\dbinom{\delta(G, V_\ell) - m+1}{2}$ edges between leader-to-follower and/or follower-to-follower vertices. However, if we need to maintain a path $\calP_i$ from a leader $\ell_x$ to a follower $f(\calD_i)$ and the path $\calP_i$ has another leader $\ell_w$ directly connected to $\ell_x$, then we will have to keep all the $\dbinom{m}{2}$ edges among the leaders as well. 
\end{proof}

In Section \ref{sec:comp}, we perform a numerical analysis and observe that the number of edges in the distance-based backbone is typically very close to the lower bound. Further, by carefully selecting the PMI sequence of the desired length, the number of edges in the distance-based backbone due to Algorithm~\ref{alg:backbone_distance} can be minimized. Next, we compare the ZFS- and distance-based backbones numerically in graphs.

\section{Comparison and Illustration}
\label{sec:comp}
The ZFS- and distance-based bounds have their own merits in computing the controllability backbones, and the number of edges in the ZFS- and distance-based backbones depends on the values of the respective bounds (as in Theorem~\ref{thm:zfs_backbone} for the ZFS backbone, and Propositions~\ref{thm:distance_lowerbound} and \ref{thm:distance_upperbound} for the distance-based backbones). As discussed previously, the ZFS bound works best when the leader set $V_\ell$ is a ZFS, implying that the network is strong structurally controllable. In such a scenario, the ZFS bound on the dimension of SSC is at least as good as the distance bound, i.e., $\zeta(G,V_\ell)\ge \delta(G,V_\ell)$~\cite{yaziciouglu2022strong}. Moreover, the number of edges in the ZFS-based backbone obtained using Algorithm~\ref{alg:backbone_zfs} will be minimum. Thus, the ZFS-based backbone is a better choice than the distance-based backbone. As an example, consider $G=(V,E)$ in Figure~\ref{fig:comp_example}(a). The leader set is $V_\ell = \{v_1,v_2,v_3\}$, which is also a ZFS. The ZFS bound is $\zeta(G,V_\ell) = 8$, which is also the dimension of SSC $\gamma(G,V_\ell)$. As for the distance-based bound, the length of the longest PMI sequence is 7, hence $\delta(G,V_\ell) = 7 < \zeta(G,V_\ell)$. Figures~\ref{fig:comp_example} (b) and (c) illustrate the resulting ZFS- and distance-based backbones, respectively. 
\begin{figure}[ht]
    \centering
     \begin{subfigure}{0.3\linewidth}
     \centering
    \includegraphics[scale=0.5]{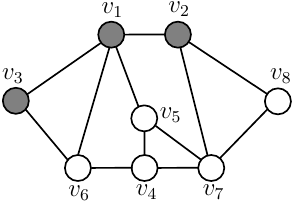}
         \caption{$G = (V,E)$}
     \end{subfigure}
             \hfill
     \begin{subfigure}{0.3\linewidth}
         \centering
\includegraphics[scale=0.5]{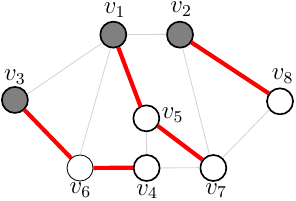}
         \caption{$B_z = (V,E_{B_z})$}
     \end{subfigure}
        \hfill
     \begin{subfigure}{0.3\linewidth}
   \centering      \includegraphics[scale=0.5]{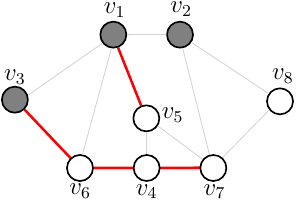}
         \caption{$B_d = (V,E_{B_d})$}
     \end{subfigure}
    \caption{(a) A graph $G$ with leaders constituting a ZFS. (b)~A ZFS-based backbone, and (c) a distance-based backbone.}
    \label{fig:comp_example}
\end{figure}
On the other hand, if the leader set is not a ZFS, the distance-based bound is typically superior to the ZFS bound \cite{yaziciouglu2022strong}. Hence, in such cases, the distance-based backbone is better than the ZFS backbone as the subgraphs that include the distance-based backbone will have higher controllability (i.e., greater dimension of SSC) than those containing the ZFS backbone. For instance, consider the same $G$ as in Figure~\ref{fig:comp_example}(a), but with a leader set $V_\ell = \{v_1,v_2\}$. Here, the ZFS bound is $\zeta(G,V_\ell) = 2$, whereas the distance bound is $\delta(G,V_\ell) = 5$. Hence, the distance-based backbone is preferable to the ZFS-based backbone since the subgraphs including the distance-based backbone have a guaranteed dimension of SSC of at least $5$. The ZFS- and distance-based backbones with $V_\ell = \{v_1,v_2\}$ are shown in Figures~\ref{fig:comp_example_2}, respectively. Note that the edge set in the ZFS-based backbone is empty. 
\begin{figure}[ht]
    \centering
     \begin{subfigure}{0.49\linewidth}
       \centering
\includegraphics[scale=0.5]{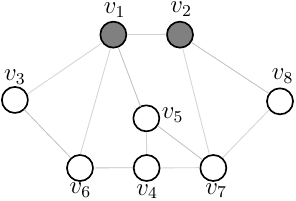}
         \caption{$B_z = (V,E_{B_z})$}
     \end{subfigure}
        \hfill
     \begin{subfigure}{0.49\linewidth}
       \centering
         \includegraphics[scale=0.5]{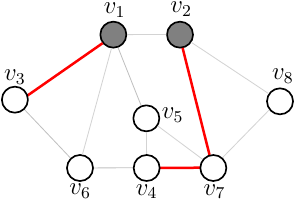}
         \caption{$B_d = (V,E_{B_d})$}
     \end{subfigure}
    \caption{(a) A ZFS-based backbone, and (b) a distance-based backbone with leaders $V_\ell = \{v_1,v_2\}$.}
    \label{fig:comp_example_2}
\end{figure}

Finally, we perform a numerical evaluation of the ZFS and distance backbones on Erd\H{o}s-R\'enyi (ER) graphs with $n = 50$ vertices and varying average density $p$. For each graph $G$, we randomly select $m = 12$ leader vertices and find a derived set and maximum length PMI sequence for the selected leader set. We use Algorithms \ref{alg:backbone_zfs} and \ref{alg:backbone_distance} to find the respective backbone graphs for $50,000$ different ER graphs for each value of $p$. Figure \ref{fig:simulation}(a) plots the average value of the ZFS and distance bounds on the dimension of SSC as a function of $p$. We observe that the distance bound $\delta(G)$ is significantly better than the ZFS bound $\zeta(G)$. Figure~\ref{fig:simulation}(b) plots the number of edges $E_{B_z}$ in the ZFS-based backbone (computed using Algorithm~\ref{alg:backbone_zfs}) as a function of $p$. Similarly, Figure~\ref{fig:simulation}(c) plots the the number of edges $E_{B_d}$ in the distance-based backbone (using Algorithm~\ref{alg:backbone_zfs}). The plot also shows the lower and upper bounds on the number of edges in the distance-based backbone as described in Propositions~\ref{thm:distance_lowerbound} and ~\ref{thm:distance_upperbound}, respectively. We observe that the number of edges in distance-based backbones is much closer to the lower bound, $\delta(G) - m$. For instance, for $p = 0.1$, the upper bound on $E_{B_d}$ is $276$ compared to the lower bound value of 23. However, the actual value of  $|E_{B_d}| \approx 29$, is much closer to the lower bound.

\vspace{-0.12in}
\begin{figure}[ht]
    \centering
     \begin{subfigure}{0.50\linewidth}
       \centering
    \includegraphics[scale=0.19]{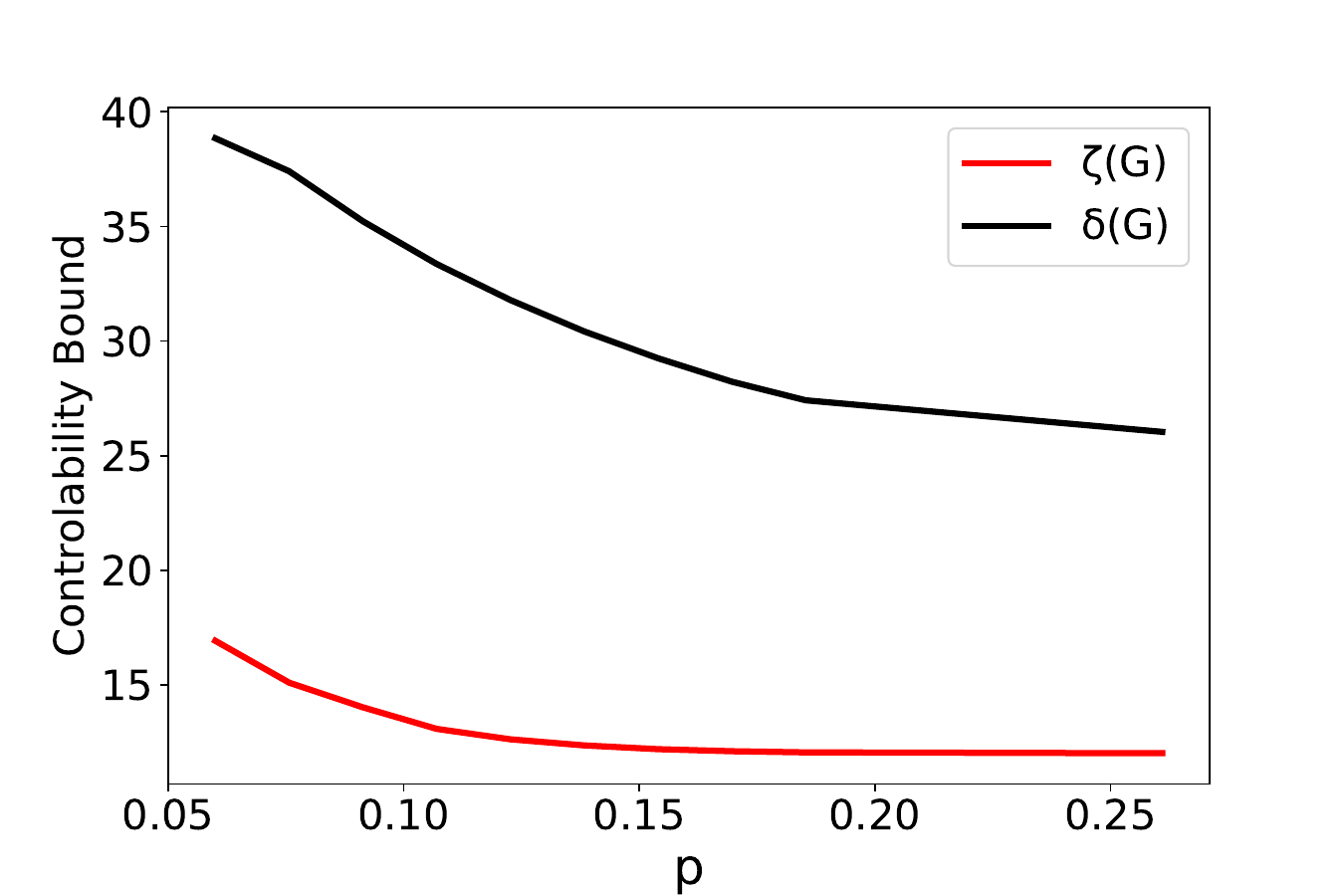}
         \caption{}
     \end{subfigure}\\
     \begin{subfigure}{0.49\linewidth}
       \centering
         \includegraphics[scale=0.18]{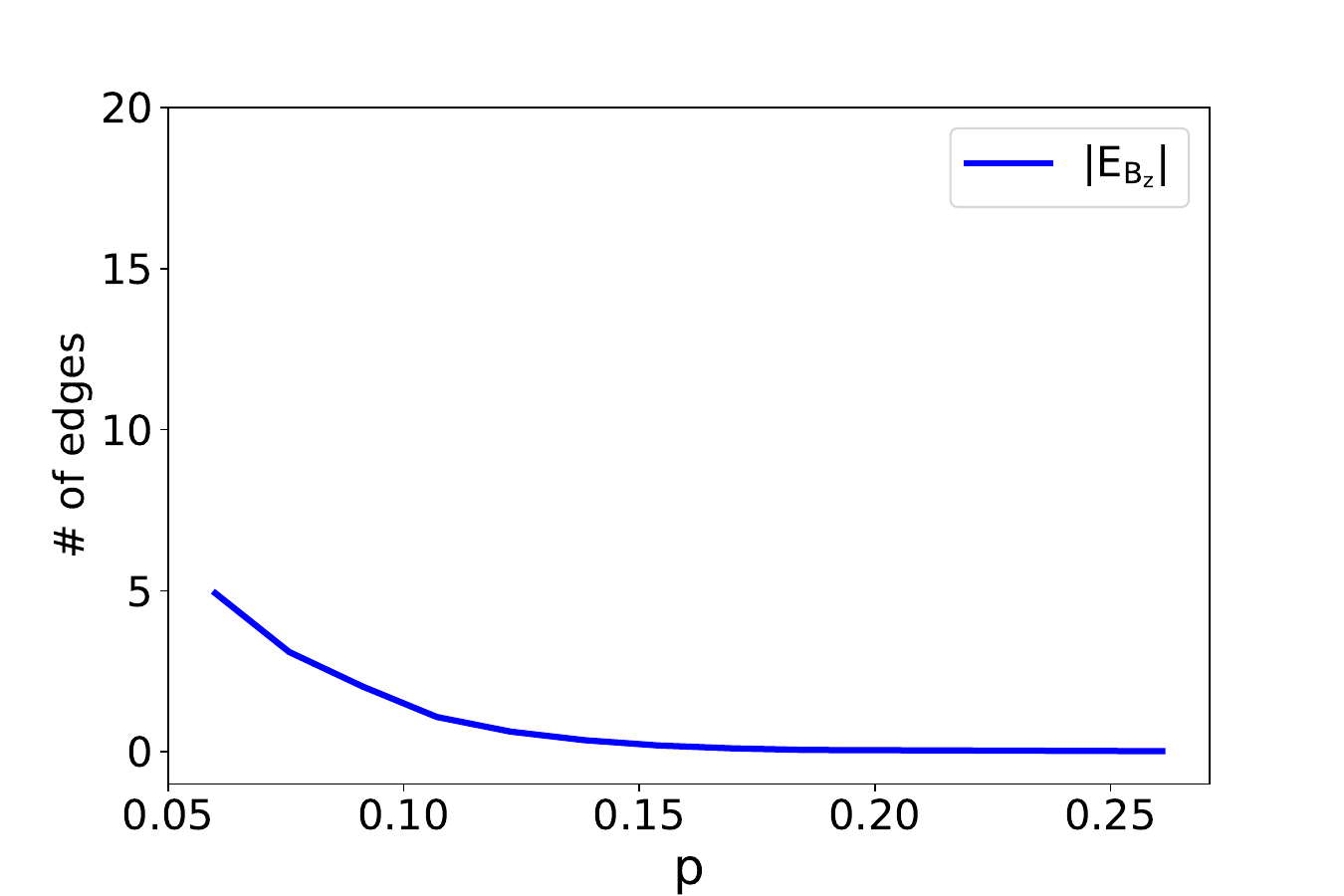}
         \caption{}
     \end{subfigure}
        \hfill
     \begin{subfigure}{0.49\linewidth}
       \centering
         \includegraphics[scale=0.18]{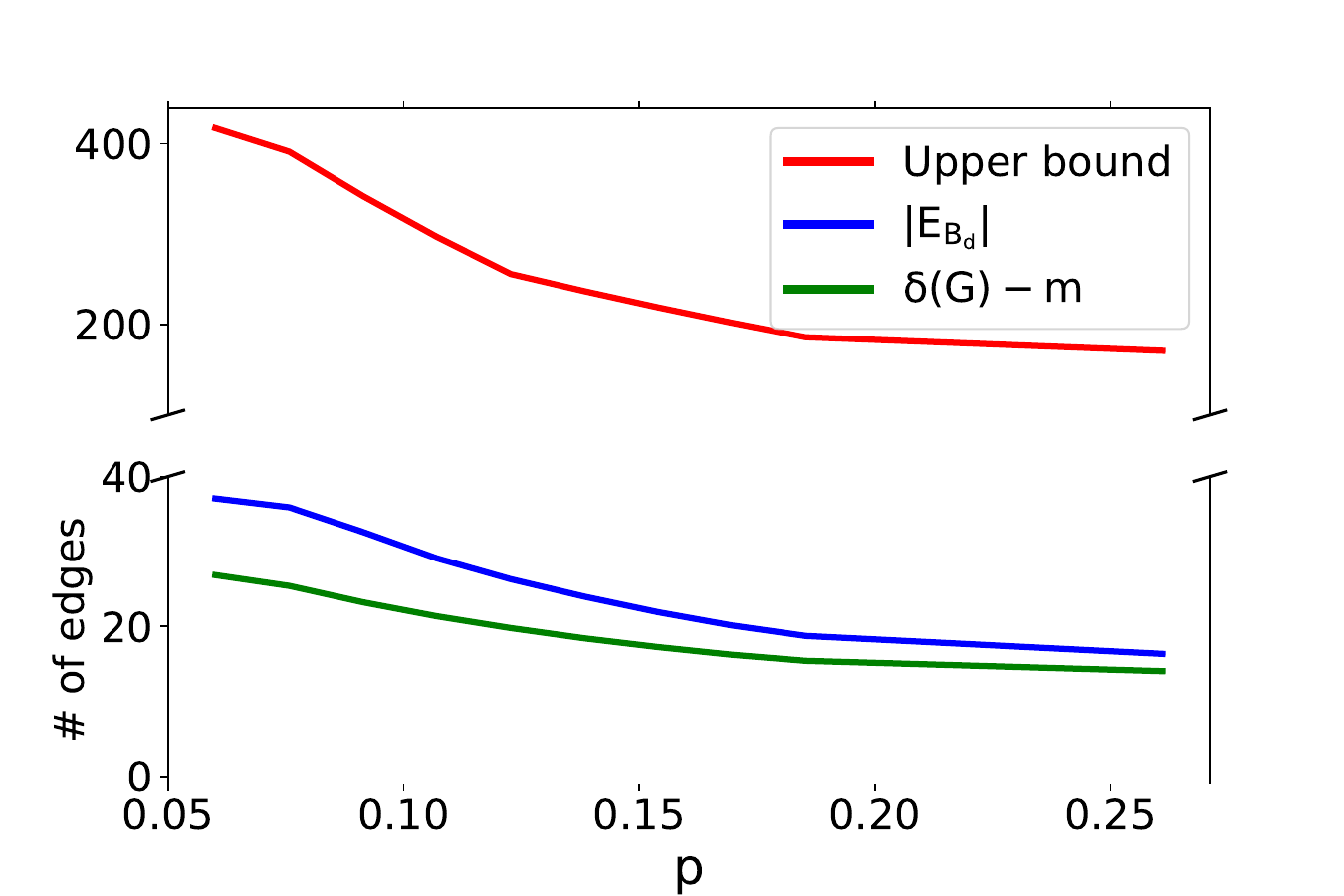}
         \caption{}
     \end{subfigure}
    \caption{(a) Comparison of the ZFS and distance bounds on SSC. (b) Number of edges in the ZFS-based backbones. (c)~Number of edges in the distance-based backbones.}
    \label{fig:simulation}
\end{figure}

\vspace{-0.15in}
\section{Conclusion}
\label{sec:conclusion}
The controllability of a network can be compromised due to changes in the network's connections. To address this issue, we proposed a method to identify a subset of edges, referred to as the "backbone edges," that are crucial for preserving the minimum network controllability even when the connections in the network are perturbed. Specifically, we designed an algorithm to compute an optimal backbone (containing the minimum number of edges) when the leader set is a zero forcing set. Moreover, we presented an algorithm utilizing the distance-based bound on the network SSC to identify a controllability backbone when the leader set is not a zero forcing set. Finally, we conducted numerical evaluations on random graphs to demonstrate the effectiveness of each algorithm. As a future direction, we plan to explore the possibility of backbone identification in networks that preserve other network properties, such as energy-based controllability parameters.

\bibliographystyle{IEEEtran}
\bibliography{refer}

\end{document}